%% file: paper-part2-slow_fading.tex
\documentclass[10pt,journal,english,twocolumn]{IEEEtran}
\usepackage{blindtext}

\usepackage{babel}
\usepackage[babel]{microtype}
\usepackage[T1]{fontenc}

\usepackage{graphicx}
\usepackage{xcolor}
\usepackage{tikz}
\usepackage{pgfplots}

\usepackage{tabularx}
\usepackage{booktabs}
\usepackage{stfloats} %

\usepackage{amsmath}
\usepackage{amsfonts}
\usepackage{amssymb}
\usepackage{amsthm}
\usepackage{bm}
\usepackage{siunitx}

\usepackage{csquotes}
\usepackage[backend=biber,url=false,style=ieee,isbn=false,doi=false]{biblatex}
\bibliography{literature.bib}
\renewcommand*{\bibfont}{\footnotesize}

\usepackage[breaklinks]{hyperref}
\hypersetup{
	pdfauthor={Karl-Ludwig Besser},
	pdftitle={Copula-Based Bounds for Multi-User Communications -- Part II: Outage Performance},
	colorlinks=true,
	linkcolor=black,
	citecolor=black,
	allcolors=.,
	urlcolor=blue,
}

\usepackage[acronyms,nomain,xindy]{glossaries}
\makeglossaries
\loadglsentries{acronyms.tex}
\setacronymstyle{long-short}

\input{setups.tex}
\title{Copula-Based Bounds for Multi-User Communications -- Part II: Outage Performance}
\author{Karl-Ludwig Besser, \IEEEmembership{Student Member, IEEE} and Eduard A. Jorswieck, \IEEEmembership{Fellow, IEEE}
\thanks{The authors are with the Institute of Communications Technology, Technische Universit\"at Braunschweig, 38106 Braunschweig, Germany (email: \{{k.besser}, {e.jorswieck}\}@tu-bs.de).}
\thanks{This work is supported in part by the German Research Foundation (DFG) under grant JO\,801/23-1.}
}
\IEEEspecialpapernotice{(Invited Paper)}

\begin{document}
\maketitle

\begin{abstract}%
	In the first part of this two-part letter, we introduced methods to study the impact of dependency on the expected value of functions of two random variables. In this second part, we present tools to derive worst- and best-case bounds on the outage probability of multi-user communication systems, including multiple access channels, wiretap channels, and reconfigurable intelligent surface-assisted channels.
\end{abstract}


\input{2-introduction.tex}
\input{2-probability-bounds-theory.tex}
\input{2-probability-bounds-applications.tex}
\input{2-conclusion.tex}

\appendices
\input{2-calculation-sum-exp.tex}

\printbibliography
\end{document}

%% file: acronyms.tex
\newacronym{5g}{5G}{Fifth-Generation Wireless Systems}
\newacronym{ask}{ASK}{amplitude-shift keying}
\newacronym[plural={AWGN}, \glsshortpluralkey={AWGN}]{awgn}{AWGN}{additive white Gaussian noise}

\newacronym{ber}{BER}{bit error rate}
\newacronym{bler}{BLER}{block error rate}
\newacronym{bpsk}{BPSK}{binary phase-shift keying}

\newacronym{cdf}{CDF}{cumulative distribution function}
\newacronym{cnn}{CNN}{convolutional neural network}
\newacronym{csi}{CSI}{channel state information}
\newacronym{csit}{CSI-T}{channel state information at the transmitter}

\newacronym{ecc}{ECC}{error-correcting code}

\newacronym{ff}{FF}{feed-forward}
\newacronym{ffnn}{FF-NN}{feed-forward neural network}
\newacronym{fsk}{FSK}{frequency-shift keying}

\newacronym{gm}{GM}{Gaussian mixture}

\newacronym{ic}{IC}{interference channel}
\newacronym{iid}{i.\,i.\,d.}{independent and identically distributed}
\newacronym{il}{IL}{information leakage}
\newacronym{iot}{IoT}{internet of things}

\newacronym{lb}{LB}{lower bound}
\newacronym{los}{LOS}{line-of-sight}

\newacronym{mac}{MAC}{multiple access channel}
\newacronym{map}{MAP}{maximum a posteriori}
\newacronym{ml}{ML}{machine learning}
\newacronym{mop}{MOP}{multi-objective programming problem}
\newacronym{mse}{MSE}{mean squared error}

\newacronym{nlos}{NLOS}{non-line-of-sight}
\newacronym{nn}{NN}{neural network}
\newacronym{nve}{NVE}{normalized validation error}

\newacronym{pdf}{PDF}{probability density function}
\newacronym{pwc}{PWC}{polar wiretap code}

\newacronym{qam}{QAM}{quadrature amplitude modulation}

\newacronym{rbf}{RBF}{radial basis function}
\newacronym{relu}{ReLU}{rectified linear unit}
\newacronym{ris}{RIS}{reconfigurable intelligent surface}
\newacronym{rnn}{RNN}{recurrent neural network}

\newacronym{sgd}{SGD}{stochastic gradient descent}
\newacronym{sgp}{SGP}{secure goodput}
\newacronym{sinr}{SINR}{signal-to-interference-plus-noise ratio}
\newacronym{snr}{SNR}{signal-to-noise ratio}
\newacronym{svm}{SVM}{support vector machine}

\newacronym{tin}{TIN}{treating interference as noise}
\newacronym{tvd}{TVD}{total variation distance}

\newacronym{uav}{UAV}{unmanned aerial vehicle}
\newacronym{ub}{UB}{upper bound}
\newacronym{urllc}{URLLC}{ultra-reliable low latency communication}

%% file: setups.tex
\pgfplotsset{compat=newest}
\pgfplotsset{plot coordinates/math parser=false}

\usetikzlibrary{patterns,arrows,positioning,external}
\usepgfplotslibrary{patchplots}

\newcommand{\abs}[1]{\ensuremath{\left|#1\right|}}
\newcommand{\positive}[1]{\ensuremath{\left[#1\right]^{+}}}
\newcommand{\leqone}[1]{\ensuremath{\left[#1\right]^{\leq 1}}}
\newcommand{\imag}{\ensuremath{\mathrm{i}}}

\newcommand{\X}{\ensuremath{\bm{X}}}

\newcommand{\Y}{\ensuremath{\bm{Y}}}

\newcommand{\Z}{\ensuremath{\bm{Z}}}

\newcommand{\lx}{\ensuremath{{\lambda_{x}}}}
\newcommand{\ly}{\ensuremath{{\lambda_{y}}}}

\renewcommand{\Pr}{\ensuremath{\mathbb{P}}}

\theoremstyle{plain}%
\newtheorem{thm}{Theorem}%

\theoremstyle{definition}
\newtheorem{defn}{Definition}%

\theoremstyle{remark}

\newtheorem*{rem*}{Remark}
\newtheorem{example}{Example}

\definecolor{plot2}{HTML}{127933}
\definecolor{plot1}{HTML}{242bb3}
\definecolor{plot0}{HTML}{F04D4A}
\definecolor{plot3}{HTML}{cc9f2b}
\definecolor{plot4}{HTML}{27AFB3}

\definecolor{change}{HTML}{0096b8}%

%% file: 2-introduction.tex
\section{Introduction}

In the first part of this two-part letter \cite{Besser2020part1}, we have presented basic tools, methods, and results to study the {expected value} of functions of dependent random variables. In this second part, we  {introduce methods} to provide bounds on the outage performance of state-of-the-art multi-user communication systems. The letter is organized as follows: first, bounds on {the joint probability} based on {copula theory} are {provided}. In Section~\ref{sec:bounds-outage-probability}, results on the outage performance {of different multi-user communication systems} are derived. The letter is concluded in Section~\ref{sec:conclusion}.

%% file: 2-probability-bounds-theory.tex
\section{Bounds on the Joint Probability}\label{sec:probability-bounds}

Copulas play a central role in the theory of joint distributions. We will therefore start this section with a brief introduction and some important facts about them.
Afterwards, we will introduce applications of copula theory for bounds on the probability of a function of random variables, e.g., the sum of two random variables.
{The presented methods will be illustrated with some simple examples.}

\subsection{Copulas}\label{sub:copulas}
One of the main applications of copulas is in the area of finance and risk management~\cite{Embrechts2002,McNeil2015}. However, we will see in {Section~\ref{sec:bounds-outage-probability}} of this work that they are also very useful in communications to derive bounds on various performance metrics, e.g., the outage probability.

All of the basic properties of copulas, which are presented in the following, can be found in \cite{Nelsen2006}.
We start with the definition of a two-dimensional copula.
\begin{defn}[{Two-dimensional Copula~\cite[Def.~2.2.2]{Nelsen2006}}]\label{def:copula}
	A \emph{two-dimensional copula} is a function $C: [0, 1]^2\to[0, 1]$ with the following properties
	\begin{enumerate}
		\item For every $a, b\in[0, 1]$
		\begin{equation*}
			C(a, 0) = 0 = C(0, b)
		\end{equation*}
		and
		\begin{equation*}
		C(a, 1) = a \quad\text{and}\quad C(1, b) = b\,;
		\end{equation*}
		\item For every $a_1, a_2, b_1, b_2\in[0, 1]$ such that $a_1\leq a_2$ and $b_1\leq b_2$
		\begin{equation*}
		C(a_2, b_2) - C(a_2, b_1) - C(a_1, b_2) + C(a_1, b_1) \geq 0\,.
		\end{equation*}
	\end{enumerate}
\end{defn}

From this, we can see that every copula is a distribution function with standard uniform marginals. The following theorem is one of the central theorems of copula theory. It shows how copulas connect the marginal distributions to a joint distribution and therefore specify a dependency structure between the marginals.

\begin{thm}[{Sklar's Theorem~\cite[Thm.~2.3.3]{Nelsen2006}}]\label{thm:sklar}
	Let $H$ be a joint distribution function with margins $F_{\X}$ and $F_{\Y}$. Then there exists a copula $C$ such that for all $x, y\in\bar{\mathbb{R}}$
	\begin{equation}\label{eq:sklar-joint-dist}
	H(x, y) = C(F_{\X}(x), F_{\Y}(y))\,.
	\end{equation}
	If $F_{\X}$ and $F_{\Y}$ are continuous, then $C$ is unique. Conversely, if $C$ is a copula and $F_{\X}$ and $F_{\Y}$ are distribution functions, then $H$ defined by \eqref{eq:sklar-joint-dist} is a joint distribution function with margins $F_{\X}$ and $F_{\Y}$.
\end{thm}

Next, we present two important copulas which act as bounds on all copulas and are often referred to as Fr\'{e}chet-Hoeffding bounds. Another important copula is the product copula $\Pi(a, b)=ab$ which corresponds to independent marginals.

\begin{thm}[{Fr\'{e}chet-Hoeffding Bounds~\cite[Thm.~2.2.3]{Nelsen2006}}]\label{thm:frechet-hoeffding-bounds}
	Let $C$ be a copula. Then for every $a, b\in[0, 1]$
	\begin{equation*}
		W(a, b) \leq C(a, b) \leq M(a, b)
	\end{equation*}
	with the copulas
	\begin{align}
		W(a, b) &= \max\left[a+b-1, 0\right]\\
		M(a, b) &= \min\left[a, b\right]\,.
	\end{align}
\end{thm}

Random variables following the Fr\'{e}chet-Hoeffding lower and upper bounds $W$ and $M$ are called countermonotonic and comonotonic random variables, respectively~\cite[Sec.~2.5]{Nelsen2006}.

Note that all of the above definitions and relations can be extended to the $n$-dimensional case where $n>2$. However, in this work, we will focus on the basic two-dimensional scenario.

\subsection{Bounds on Functions of Random Variables}
In communication {scenarios with slow fading channels}, we {often} face the situation that performance metrics, e.g., the channel capacity, are random variables due to the random nature of involved variables like the channel gain. Therefore, a probabilistic view of these quantities is of interest. One example, which gained more attention recently, is the probability distribution of these metrics, e.g., in the context of \gls{urllc}~\cite{Bennis2018}.

In the following, we will present bounds on the probability of random variables which are formed by certain binary operations of two random variables. This was first introduced for the sum of two random variables in \cite{Frank1987} but it is easily extended to more general functions~\cites[Sec.~5]{Frank1987}[Thm.~1]{Williamson1990}. {In communications, this can be used to bound the outage probability for systems with slow fading channels.}

\begin{thm}[{\cite[Thm.~1]{Williamson1990}}]\label{thm:bounds-sum-frank}
	Let $\X$ and $\Y$ be random variables over the non-negative real numbers with \glspl{cdf} $F_{\X}$ and $F_{\Y}$, respectively. Let $L$ be a binary operation that is non-decreasing in each place and continuous. The \gls{cdf} of the random variable $\Z=L(\X, \Y)$ is bounded by
	\begin{equation}\label{eq:bounds-dist-z}
	\tau_{W}(F_{\X}, F_{\Y}) \leq F_{\Z} \leq \phi_{W}(F_{\X}, F_{\Y})\,,
	\end{equation}
	with
	\begin{align}
	\tau_{C}(F_{\X}, F_{\Y})(s) &= \sup_{L(x, y)=s} C(F_{\X}(x), F_{\Y}(y))\label{eq:lower-bound-joint-cdf}\\
	\phi_{C}(F_{\X}, F_{\Y})(s) &= \inf_{L(x, y)=s} \bar{C}(F_{\X}(x), F_{\Y}(y))\label{eq:upper-bound-joint-cdf}
	\end{align}
	for a copula $C$ and its dual $\bar{C}(a, b)=a+b-C(a, b)$.
\end{thm}
\begin{proof}
	We will only give an intuition on the proof in the following. The details can be found in \cite{Frank1987} and \cite{Williamson1990}.
	
	Imagine the line corresponding to $L(x, y)=s$, e.g., $x+y=s$. The probability to bound is given by the probability mass distributed below this line, e.g., $x+y\leq s$. Now pick any two points $(x_1, y_1)$ and $(x_2, y_2)$ on the line $L(x, y)=s$ and observe the following. First, the probability of interest is greater or equal to the probability mass in the area $\{(x, y)\;|\; x\leq x_1 \wedge y\leq y_1\}$ but less or equal to the probability in the area $\{(x, y)\;|\; x\geq x_2 \wedge y\geq y_2\}$. Combining this with the Fr\'{e}chet-Hoeffding bounds from Theorem~\ref{thm:frechet-hoeffding-bounds} gives \eqref{eq:bounds-dist-z}.
\end{proof}

The bounds from Theorem~\ref{thm:bounds-sum-frank} are point-wise {tight}. This means that {at each point $s$, there exist two joint distributions which achieve the lower and upper bound, respectively.} However, the joint distribution may be different for every point $s$. Details about copulas which achieve the bounds can be found in \cite[Thm.~3.2]{Frank1987} and \cite[Thm.~3]{Williamson1990}.
We only want to give the intuition behind the construction of the dependency structure. For this purpose, we now consider the lower bound. In this case, the goal is to minimize the probability mass of the joint distribution in the area $L(x, y)<s$. This is achieved by splitting the copula in a comonotonic and a countermonotonic part such that the mass equivalent to the lower bound is placed in the area $L(x, y)<s$.
An illustration of this can be found in Fig.~\ref{fig:joint-pdf-sum-uniform-lower}.

All of the above results are for the case of two random variables $\X$ and $\Y$. There are some results for the case of more than two random variables. In~\cite{Wang2013}, the authors derive bounds on the probability of the sum of $n$ random variables. However, there are some restrictions on the distributions of the variables. An application of this in the context of communications can be found in \cite{Besser2019}.

\subsection{Example {-- Uniform Marginals}}\label{sub:example-uniform}
In the following, we will give two simple examples to illustrate Theorem~\ref{thm:bounds-sum-frank}. For the marginals, we assume uniform distributions, i.e., $\X\sim\mathcal{U}[x_{\text{min}}, x_{\text{max}}]$ and $\Y\sim\mathcal{U}[y_{\text{min}}, y_{\text{max}}]$.

\begin{example}
	As a first example function, we use the sum, i.e., $L(x, y)=x+y$. This is the example considered in \cite[Sec.~4]{Frank1987}.
	For the lower bound $\underline{F}_{\X+\Y}(s)$, we have to solve the problem
	\begin{equation*}
	\underline{F}_{\X+\Y}(s) = \sup_{x+y=s} \positive{F_{\X}(x) + F_{\Y}(y) - 1}\,.
	\end{equation*}
	This gives a uniform distribution of $\X+\Y$ between $\min\left[x_{\text{min}}+y_{\text{max}}, y_{\text{min}}+x_{\text{max}}\right]$ and $x_{\text{max}}+y_{\text{max}}$.
	
	Similarly, the upper bound is derived as
	\begin{equation*}
	\overline{F}_{\X+\Y}(s) = \inf_{x+y=s} \leqone{F_{\X}(x) + F_{\Y}(y)}\,,
	\end{equation*}
	which gives a uniform distribution of $\X+\Y$ between $x_{\text{min}}+y_{\text{min}}$ and $\max\left[x_{\text{min}}+y_{\text{max}}, y_{\text{min}}+x_{\text{max}}\right]$.
	Figure~\ref{fig:example-uniform} shows the lower and upper bound on the distribution of $\X+\Y$ with $\X\sim\mathcal{U}[1, 3]$ and $\Y\sim\mathcal{U}[2, 5]$.
	This can also be found as an interactive version at \cite{BesserGitlab}.
	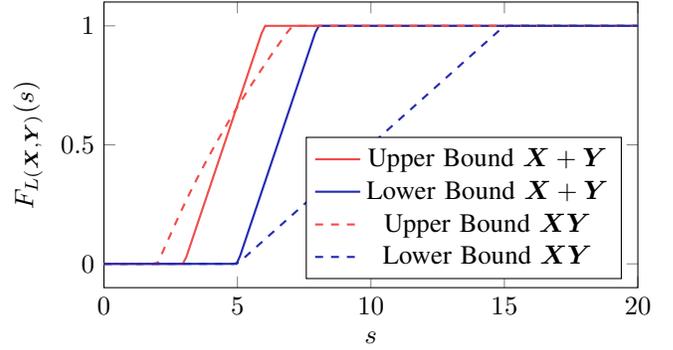
\begin{figure}
		\centering
		\input{img/example-uniform.tex}
		
		\vspace*{-1em}
		\caption{Upper and lower bounds on the probability of $\X+\Y<s$ and $\X\Y<s$ with $\X\sim\mathcal{U}[1, 3]$ and $\Y\sim\mathcal{U}[2, 5]$.}
		\label{fig:example-uniform}
	\end{figure}
	As described previously, the joint distributions, for which the bounds are achieved, vary for different $s$. An example of the joint \gls{pdf} for the lower bound on the sum of uniform distributions can be found in Fig.~\ref{fig:joint-pdf-sum-uniform-lower}. Additionally, the line $x+y=s$ is shown. This illustrates the intuition behind the construction of the joint distribution. The probability mass in the area $x+y<s$ is minimized and corresponds to the lower bound on the probability.
	\begin{figure}
		\centering
		\input{img/joint-pdf-sum-uniform-lower.tex}
		
		\vspace*{-0.8em}
		\caption{The black lines indicate the support of the joint \gls{pdf} $f_{\X, \Y}$ that achieves the lower bound for the sum of $\X\sim\mathcal{U}[1, 3]$ and $\Y\sim\mathcal{U}[2, 5]$ with $s=6$. The dashed curve corresponds to the line $x+y=s$.}
		\label{fig:joint-pdf-sum-uniform-lower}
	\end{figure}
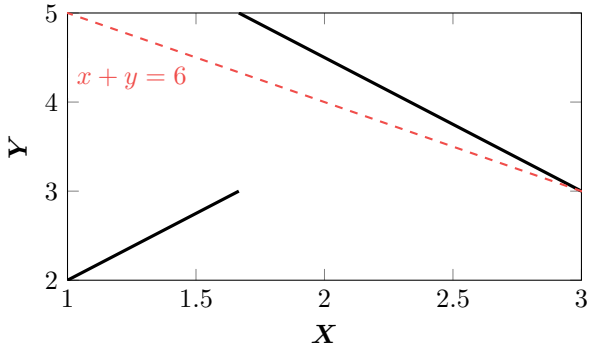
\end{example}

\begin{example}
	The next example is the product of $\X$ and $\Y$, i.e., $L(x, y)=xy$. This function is also investigated in \cite{Williamson1990}.
	The optimization problems are similar. The only difference is that they are performed over the set $xy=s$ instead of $x+y=s$.
	For the example $\X\sim\mathcal{U}[1, 3]$ and $\Y\sim\mathcal{U}[2, 5]$, the bounds are also shown in Fig.~\ref{fig:example-uniform}. They can also be found as interactive versions at \cite{BesserGitlab}.
\end{example}

%% file: img/example-uniform.tex
\begin{tikzpicture}%
\begin{axis}[
	width=.98\linewidth,
	height=.22\textheight,
	xmin=0,
	xmax=20,
	domain=0:20,
	xlabel={$s$},
	ylabel={$F_{L(\X,\Y)}(s)$},
	legend pos=south east,
]
\addplot[thick, plot0, samples=150] {max(min((x-3)/(6-3), 1), 0)}; %
\addlegendentry{Upper Bound $\X+\Y$};

\addplot[thick, plot1, samples=150] {max(min((x-5)/(8-5), 1), 0)}; %
\addlegendentry{Lower Bound $\X+\Y$};

 X~U[1, 3] and Y~U[2, 5]
\addplot[thick, plot0, dashed] table[x=s, y=upper] {data/product_uniform_upper-X1_3-Y2_5.dat};
\addlegendentry{Upper Bound $\X\Y$};

\addplot[thick, plot1, samples=100, dashed] {max(min((x-5)/(15-5), 1), 0)}; %
\addlegendentry{Lower Bound $\X\Y$};
\end{axis}

\end{tikzpicture}

%% file: img/joint-pdf-sum-uniform-lower.tex
\begin{tikzpicture}
\begin{axis}[
	width=.95\linewidth,
	height=.21\textheight,
	xlabel={$\X$},
	ylabel={$\Y$},
	view={0}{90},
	xmin=1,
	xmax=3,
	ymin=2,
	ymax=5,
]
    \addplot[black, very thick] coordinates {(1, 2) (1.66666, 3)};
    \addplot[black, very thick] coordinates {(1.66666, 5) (3, 3)};
    \addplot[plot0, thick, dashed] coordinates {(0, 6) (6, 0)};
    \node[anchor=north east, text=plot0, fill=white] at (axis cs: 1.5, 4.5) {$x+y=6$};
\end{axis}
\end{tikzpicture}

%% file: 2-probability-bounds-applications.tex
\section{Bounds on the Outage Probability}\label{sec:bounds-outage-probability}

For slow fading channels, the outage probability and outage capacity region are the correct performance measures. 
Given an achievable rate $R(x)$ as a function of the random channel realization $x$, the outage probability is defined as the probability that the next channel realization $x$ results in an achievable rate lower than the transmission rate $R$, i.e.,
\begin{equation}
\Pr\left( R(\X) < R \right) \label{eq:out1}.
\end{equation} 
Based on the results from {Section~\ref{sec:probability-bounds}} and \cite{Frank1987}, we derive new lower and upper bounds on the outage probabilities for three modern communication scenarios, including the slow fading \gls{mac}, slow fading wiretap channels, and the \gls{ris} assisted slow fading channel. 
{In the following, we assume statistical \gls{csit} and perfect CSI at the receiver.}

{
\subsection{Point-to-Point Outage Probability}	
First, we will present the bounds on the outage probability for a point-to-point transmission over two slow Rayleigh fading links $h_x$ and $h_y$ and compare them to an existing correlation model from literature.
It is well-known, that the outage probability $\varepsilon$ of such a channel is given by~\cite{Tse2005}
\begin{equation}
\varepsilon = \Pr\left(\X+\Y < s\right)
\end{equation}
where we use the shorthands $\X=\abs{h_x}^2$, $\Y=\abs{h_y}^2$, and $s=\frac{2^R-1}{\xi}$ with the \gls{snr} $\xi$. For Rayleigh fading, the channel gains $\X$ and $\Y$ are exponentially distributed with shape parameters $\lx$ and $\ly$, respectively. We can now use Theorem~\ref{thm:bounds-sum-frank} to derive the general bounds on $\varepsilon$ for all possible joint distributions of $\X$ and $\Y$.}
The probability of $\X+\Y<s$ is lower bounded by a shifted exponential distribution with shape parameter $\underline{\lambda}_{x+y}=\lx\ly/(\lx+\ly)$. The upper bound on the probability is an exponential distribution with shape parameter $\overline{\lambda}_{x+y}=\min\left[\lx, \ly\right]$. The detailed calculations can be found in Appendix~\ref{sec:example-sum-exp} and online at \cite{BesserGitlab}.

{For comparison, we use the correlation model from \cite{Chen2004,Dang2017} where the channel coefficients $h_x$ and $h_y$ are given as
\begin{equation}
h_i = \left(\sqrt{1-\rho}x_i+\sqrt{\rho}x_{0}\right) + \imag \left(\sqrt{1-\rho}y_i+\sqrt{\rho}y_{0}\right)\,,
\end{equation}
where $\imag^2=-1$, $x_i, y_i\sim\mathcal{N}(0, 1/2)$ are \gls{iid}, and $x_0, y_0\sim\mathcal{N}(0, 1/2)$ are \gls{iid} and used as reference to correlate. The correlation coefficient is denoted by $\rho$. The calculations are based on \cite{Chen2004,Dang2017} and can be found online at~\cite{BesserGitlab}.
Figure~\ref{fig:comparison-literature-model} shows the outage probability for different values of $\rho$ for an \gls{snr} $\xi=\SI{10}{\decibel}$ and transmission rate $R=1$. It can be seen that there is a gap of the outage probability for this particular correlation model to the general bounds from Theorem~\ref{thm:bounds-sum-frank}. This shows that the linear correlation model underestimates the worst-case and best-case outage probabilities that could occur under arbitrary dependency.
\begin{figure}
	\centering
	\input{img/comparison-rayleigh-correlation.tex}
	
	\vspace*{-0.8em}
	\caption{Comparison of the copula-based bounds on the outage probability to the correlation model from~\cite{Chen2004,Dang2017}. The \gls{snr} and transmission rate are $\xi=\SI{10}{\decibel}$ and $R=1$, respectively.}
	\label{fig:comparison-literature-model}
\end{figure}
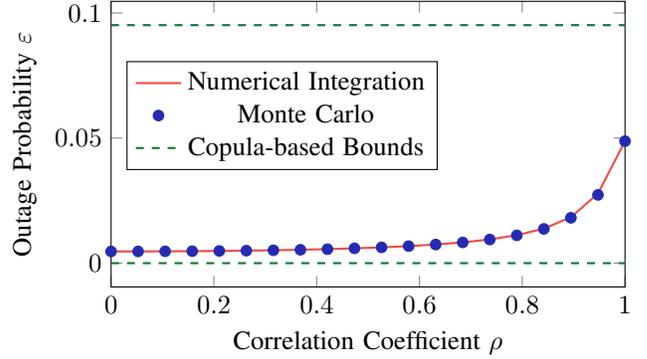
}

\subsection{MAC Outage Probability}
An outage in the two user \gls{mac} occurs, if at least one of the following events occurs
\begin{align*}
E_1:\quad &\log(1+\X)<R_1\\
E_2:\quad &\log(1+\Y)<R_2\\
E_3:\quad &\log(1+\X+\Y)<R_1+R_2\,,
\end{align*}
which is illustrated in Fig.~\ref{fig:areas-mac-outage}, where we introduce the shorthands $2^{R_1}-1 = \alpha$, $2^{R_2}-1 = \beta$, and $2^{R_1+R_2}-1 = s$.
It can easily be seen that {Theorem~\ref{thm:bounds-sum-frank}} can be applied to bound the outage probability
\begin{equation}
\varepsilon_{\text{MAC}} = \Pr\left(E_1\cup E_2\cup E_3\right)\,.
\end{equation}
When calculating the bounds based on {Theorem~\ref{thm:bounds-sum-frank}}, we need to optimize over the line $\mathcal{L}$ shown in Fig.~\ref{fig:areas-mac-outage}.
The optimization function for the lower bound {in \eqref{eq:lower-bound-joint-cdf}} is given as
\begin{equation}
g(x, y) = \positive{F_{\X}(x) + F_{\Y}(y) - 1}\,.
\end{equation}
We split $\mathcal{L}$ in three parts and optimize over them separately. The maximum values over the parts $y\geq s-\alpha$ and $x\geq s-\beta$ are given as $F_{\X}(\alpha)$ and $F_{\Y}(\beta)$, respectively.
For the middle part of $\mathcal{L}$, given by $x+y=s$, we get the optimum $y$ as
\begin{equation}
y^\star = \begin{cases}
\beta & \text{if}\quad y_{\text{opt}} \leq \beta\\
y_{\text{opt}} & \text{if}\quad \beta < y_{\text{opt}} < s-\alpha\\
s-\alpha & \text{if}\quad y_{\text{opt}} \geq s-\alpha
\end{cases}
\end{equation}
with
\begin{equation}
y_{\text{opt}} = \frac{\lx s + \log\frac{\ly}{\lx}}{\lx + \ly}\,.
\end{equation}
The minimum value is then given as $g(s-y^\star, y^\star)$.

Combining the three parts gives the final expression for the lower bound on the outage probability of the two user \gls{mac} as
\begin{equation}
\underline{\varepsilon} = \max\left[F_{\X}(2^{R_1}-1),\; g(s-y^\star, y^\star),\; F_{\Y}(2^{R_2}-1)\right]
\end{equation}

For the upper bound {in \eqref{eq:upper-bound-joint-cdf}}, we use the same general idea. However, we now want to minimize the optimization function
\begin{equation*}
h(x, y) = \leqone{F_{\X}(x) + F_{\Y}(y)}\,.
\end{equation*}
The upper bound is then given as
\begin{equation}
\overline{\varepsilon} = \min\left[{F_{\X}(\alpha) + F_{\Y}(s-\alpha), F_{\X}(s-\beta) + F_{\Y}(\beta), 1}\right].
\end{equation}

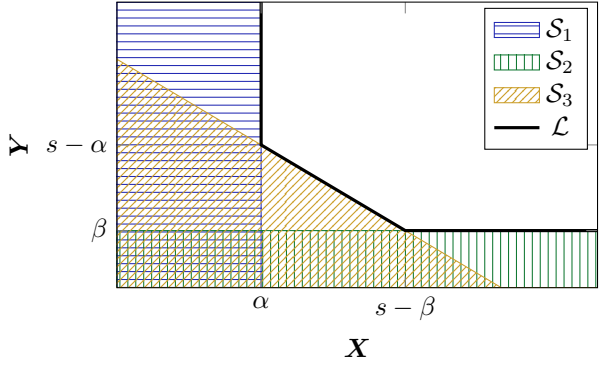
\begin{figure}
	\centering
		\input{img/areas-mac-outage.tex}

	\vspace*{-1em}
	\caption{Integration areas for the outage probability of a two user \gls{mac}. The regions $\mathcal{S}_i$ correspond to the individual events $E_i$.}
	\label{fig:areas-mac-outage}
\end{figure}

\subsection{Secrecy Outage Probability}
The outage probability in slow fading wiretap channels is an important measure to take both outages due to reliability drops as well as information leakage to the unknown eavesdropper into account. In \cite{Besser2020tcom}, we derive worst-case and best-case outage probabilities for the scenarios with perfect \gls{csit} about the legitimate channel and without this information. The channel to the eavesdropper is always unknown.

\subsection{RIS Channel}
As shown in {Section~\ref{sub:example-uniform}}, we can also apply {Theorem~\ref{thm:bounds-sum-frank}} for other functions than the sum, e.g., for the product of random variables over the non-negative reals. In wireless communications this problem arises as the product of channel gains, e.g., in the context of \gls{ris}~\cite{Bjornson2020}.
For illustration, we present a simplified example in the following. Assume that $\X$ and $\Y$ represent (dependent) channel gains of different channels which are exponentially distributed, i.e., $\X \sim \exp(\lx), \Y \sim \exp(\ly)$. The achievable rate expression for the special case without direct link $\beta_{sd}=0$ and with $N=1$ in \cite[Lemma 1, $R_{IRS}$ in (12)]{Bjornson2020} results in 
\begin{equation}
R(\X, \Y) = \log_2\left(1 + \xi \X \Y \right)\, ,
\end{equation}
with \gls{snr} $\xi$.  
The outage probability is defined according \eqref{eq:out1}. In the case of independent channels, this is given as
\begin{equation*}
\varepsilon_{\text{ind}} = 1 - 2 \sqrt{\lx \ly s}\ K_{1}\!\left(2 \sqrt{\lx \ly s}\right)\,,
\end{equation*}
where $K_{1}$ is the modified Bessel function of second kind and order $1$ \cite{Abramowitz1970} and $s=2^R-1$.
The lower and upper bounds are given according to {Theorem~\ref{thm:bounds-sum-frank}} as
\begin{align*}
\underline{\varepsilon} &= \sup_{xy=s} \positive{F_{\X}(x) + F_{\Y}(y) - 1}\\
\overline{\varepsilon} &= \inf_{xy=s} \leqone{F_{\X}(x) + F_{\Y}(y)}\,,
\end{align*}
respectively. We solve these problems numerically. 
\begin{figure}
	\centering
	\input{img/outage-probability-product-rayleigh.tex}
	
	\vspace*{-.8em}
	\caption{Bounds on the outage probability for \gls{ris} channel for \gls{snr} $\xi=\SI{0}{\decibel}$ and $\lx=\ly=1$.}
	\label{fig:outage-product-rayleigh}
	\vspace*{-0.5em}
\end{figure}
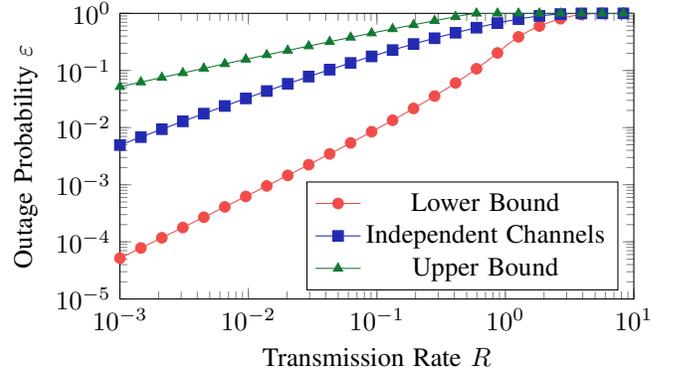
The three different cases are shown in Fig.~\ref{fig:outage-product-rayleigh} for varying values of the transmission rate $R$ and with \gls{snr} $\xi=1$. The figure, together with the numerical calculations, can be found as an interactive version at~\cite{BesserGitlab}.
It can be seen that the achievable rate depends significantly on the dependency of the two channel realizations from source to \gls{ris} and from \gls{ris} to destination. For a rate of $0.1$ bits per second, the best outage probability is below \SI{1}{\percent} while in the worst case it is about \SI{50}{\percent}.

\subsection{Extension to More Than Two Links/Antennas}
The results on the outage performance presented in this section are all based on the complete characterization of the worst- and best-case dependency in {Theorem~\ref{thm:bounds-sum-frank}}. The extension to more than two random variables, $n > 2$, is not straightforward. Mainly, because the lower Fr\'{e}chet-Hoeffding bound does not result in valid copula, i.e., a valid joint distribution \cite{Nelsen2006}. 

{However, it is possible to derive upper and lower bounds on the outage probability for $n>2$ for some special cases. In \cite{Besser2020icc}, we derive the outage probability bounds for $n$ random variables in the context of diversity. The constraint in this case is that the fading distributions have to have monotone densities. The derivation of the bounds is then based on {joint mixability}~\cite{Wang2013,Wang2016}.}
It is also possible to calculate the bounds numerically by a rearrangement algorithm~\cite{Puccetti2012}. This algorithms supports an arbitrary number of random variables with arbitrary distributions.

%% file: img/comparison-rayleigh-correlation.tex
\begin{tikzpicture}%
\begin{axis}[
	width=.95\linewidth,
	height=.22\textheight,
	xmin=0,
	xmax=1,
	xlabel={Correlation Coefficient $\rho$},
	ylabel={Outage Probability $\varepsilon$},
	yticklabel style={
		/pgf/number format/fixed,
		/pgf/number format/precision=3,
	},
	scaled y ticks=false,
	legend entries={{Numerical Integration}, {Monte Carlo}, {Copula-based Bounds}},
	legend style={at={(0.03,0.6)},anchor=west},
]
\addplot[plot0, thick] table[x=rho,y=integration] {data/rayleigh-comparison-corr-snr10-rate1.dat};
\addplot[plot1, only marks, mark=*] table[x=rho,y=mc] {data/rayleigh-comparison-corr-snr10-rate1.dat};
\addplot[plot2, thick, dashed] {0.09516258196404044};
\addplot[plot2, thick, dashed] {0.0};
\end{axis}
\end{tikzpicture}

%% file: img/areas-mac-outage.tex
\begin{tikzpicture}
\begin{axis}[
	width=.9\linewidth,
	height=.22\textheight,
	xmin=0,
	xmax=5,
	ymin=0,
	ymax=5,
	axis on top,
	xlabel={$\X$},
	ylabel={$\Y$},
	xtick={1.5, 3},
	ytick={1, 2.5},
	xticklabels={{$\alpha$},{$s-\beta$}},
	yticklabels={{$\beta$},{$s-\alpha$}},
]	
	\addplot[plot1,pattern=horizontal lines, pattern color=plot1, domain=0:1.5, area legend] {5}\closedcycle;
	\addlegendentry{$\mathcal{S}_{1}$};
	
	\addplot[plot2,pattern=vertical lines, pattern color=plot2, domain=0:5, area legend] {1}\closedcycle;
	\addlegendentry{$\mathcal{S}_{2}$};
	
	\addplot[plot3,pattern=north east lines, pattern color=plot3, area legend] {-x+4}\closedcycle;
	\addlegendentry{$\mathcal{S}_{3}$};

	\addplot[very thick, black] table {
		1.5	5
		1.5	2.5
		3	1
		5	1
	};
	\addlegendentry{$\mathcal{L}$};

\end{axis}
\end{tikzpicture}

%% file: img/outage-probability-product-rayleigh.tex
\begin{tikzpicture}%
\begin{axis}[
	width=.95\linewidth,
	height=.22\textheight,
	xlabel={Transmission Rate $R$},
	ylabel={Outage Probability $\varepsilon$},
	xmin=1e-3,
	xmax=10,
	ymin=1e-5,
	ymax=1,
	xmode=log,
	ymode=log,
	legend pos=south east,
]
\addplot[plot0, mark=*, mark repeat=2] table[x=rate,y=lower] {data/bounds-product-rayleigh-lx1-ly1.dat};
\addlegendentry{Lower Bound};

\addplot[plot1, mark=square*, mark repeat=2] table[x=rate,y=indep] {data/bounds-product-rayleigh-lx1-ly1.dat};
\addlegendentry{Independent Channels};

\addplot[plot2, mark=triangle*, mark repeat=2] table[x=rate,y=upper] {data/bounds-product-rayleigh-lx1-ly1.dat};
\addlegendentry{Upper Bound};

\end{axis}
\end{tikzpicture}

%% file: 2-conclusion.tex
\section{Conclusion}\label{sec:conclusion}

{This second part dealt with the outage probability, which is a typical performance metric used for scenarios with slow-fading channels. We considered various communication settings, e.g., the \gls{mac} and the wiretap channel, and showed that the joint distribution of the channels has a significant impact on the outage performance.

A more extensive conclusion and outlook for future work and applications can be found at the end of the first part of this letter~\cite{Besser2020part1}.}

%% file: 2-calculation-sum-exp.tex
\section{Sum of Exponentially Distributed Random Variables}\label{sec:example-sum-exp}
We start with the lower bound
\begin{equation*}
\Pr(\X+\Y<s)\leq \underline{F}_{\X, \Y}(s) = \sup_{x+y=s} \positive{F_{\X}(x)+F_{\Y}(y)-1}
\end{equation*}
which we rewrite as
\begin{equation}\label{eq:example-sum-exp-optimization-lower}
\underline{F}_{\X, \Y}(s) = \sup_{x\geq 0} \positive{g(x)}
\end{equation}
with the shorthand $g(x) = F_{\X}(x)+F_{\Y}(s-x)-1$.
First, we take a look at the boundaries
\begin{align*}
\positive{g(0)} &= \positive{F_{\Y}(s) - 1} = 0\\
\lim\limits_{x\to\infty} g(x) &= 0\,.
\end{align*}
Next, we need the first two derivatives of $g$ in order to solve the optimization problem. For the case of exponentially distributed $\X\sim\exp(\lx)$ and $\Y\sim\exp(\ly)$, these are
\begin{align*}
g^\prime(x) = \frac{\partial g}{\partial x} &= \frac{1}{\alpha}\exp\left(-\frac{x}{\alpha}\right) - \frac{1}{\beta}\exp\left(-\frac{s-x}{\beta}\right)\\
g''(x) = \frac{\partial^2 g}{\partial x^2} &= \frac{-1}{\alpha^2}\exp\left(-\frac{x}{\alpha}\right) - \frac{1}{\beta^2}\exp\left(-\frac{s-x}{\beta}\right)\,,
\end{align*}
where we introduce $\alpha=1/\lx$ and $\beta=1/\ly$.

For the stationary point $x^\star$, we have $g^\prime(x^\star)=0$. Solving this for our example gives
\begin{equation*}
x^\star = \frac{\left(\frac{s}{\beta}+\log\left(\frac{\beta}{\alpha}\right)\right)\alpha\beta}{\alpha+\beta}\,.
\end{equation*}
Since $g''(x^\star)<0$, we know that this point is a maximum, which we need for the optimization problem in \eqref{eq:example-sum-exp-optimization-lower}.
Therefore, we can now evaluate $\positive{g(x^\star)}$ in order to get the lower bound on $F_{\X+\Y}$. In our example, this is
\begin{equation*}
\positive{g(x^\star)} = \positive{1 - \exp\left(-\frac{s-k}{\alpha+\beta}\right)} = \underline{F}_{\X+\Y}(s)\,,
\end{equation*}
with $k=(\alpha+\beta)\log(\alpha+\beta) - \beta\log\beta - \alpha\log\alpha$.
This shows that the lower bound of $\X+\Y$ is a shifted exponential distribution with shape parameter $\lambda_{x+y}=1/(\alpha+\beta)$. This is confirmed by the result from \cite[Sec.~4]{Frank1987}.

For the upper bound, we have to solve the similar problem
\begin{equation*}
\overline{F}_{\X, \Y}(s) = \inf_{x\geq 0} \leqone{F_{\X}(x)+F_{\Y}(s-y)}\,.
\end{equation*}
Starting with the limits,
\begin{align*}
F_{\X}(0) + F_{\Y}(s) &= F_{\Y}(s) \leq 1\\
\lim\limits_{x\to\infty} F_{\X}(x) + F_{\Y}(s-x) &= 1\,
\end{align*}
we obtain $F_{\Y}(s)$ as an initial lower bound. Next, we have another lower bound at $x=s$ which is equal to $F_{\X}(s)$. In addition, we observe that $F_{\X}(x)+F_{\Y}(s-x)>F_{\X}(s)$ for $x>s$.
Since the derivatives are equivalent to the ones of $g$, we know that the optimization function can only have a maximum in $0<x<s$. Combining all of this, we get $\overline{F}_{\X+\Y}(s) = \min\left[F_{\X}(s), F_{\Y}(s)\right]$,
which corresponds to an exponential distribution with shape parameter $\lambda_{x+y}=1/(\max\left[\alpha, \beta\right])$.